\documentclass[11pt]{amsart}
\usepackage{graphicx}
\usepackage{amsmath,amssymb,amsthm}
\usepackage{amssymb}
\usepackage{amsthm}
\usepackage{hyperref}
\usepackage{color}

\newtheorem{theorem}{Theorem}
\newtheorem{lemma}[theorem]{Lemma}
\newtheorem{claim}[theorem]{Claim}
\newtheorem{corollary}[theorem]{Corollary}
\newtheorem{remark}[theorem]{Remark}

\newtheorem {conjecture}[theorem]{Conjecture}

\newcommand{\R}{{\mathbb R}}

\newcommand{\fs}{{\mathrm{fs}}}

\newcommand{\F}{{\mathcal{F}}}
\newcommand{\nb}{{\mathrm{NB}}}

\DeclareMathOperator{\STAB}{STAB}

\DeclareMathOperator{\QSTAB}{CLIQUE}

\DeclareMathOperator{\conv}{conv}
\DeclareMathOperator{\cone}{cone}

\DeclareMathOperator{\diag}{diag}
\newcommand{\one}{\mathbf{1}} 

\newcommand{\set}[1]{\left\{#1\right\}}

\DeclareMathOperator{\fra}{\textup{FRAC}}
\DeclareMathOperator{\stab}{\textup{STAB}}
\DeclareMathOperator{\oddcycle}{\textup{OC}}

\DeclareMathOperator{\clique}{\textup{CLIQUE}}

\DeclareMathOperator{\NB}{\textup{NB}}
\DeclareMathOperator{\LS}{\textup{LS}}

\def\diag{\textup{diag}}
\def\th{\textup{TH}}

\def\R{{\mathbb R}}
\def\S{{\mathbb S}}

\def\e{{\mathbf e}}

\begin{document}

\title[SDP-Operator, Near-Perfect and Near-Bipartite Graphs]{Lov\'{a}sz-Schrijver SDP-operator,\\
near-perfect graphs and
near-bipartite graphs}\thanks{Some of the results in this paper were first announced in conference
proceedings in abstracts \cite{BENT2,BENT3}.  This research was supported in part by grants
PID-CONICET 241, PICT-ANPCyT 0361, ONR N00014-12-10049, and a Discovery Grant from NSERC}

\author{S. Bianchi, M. Escalante, G. Nasini, L. Tun\c{c}el}\thanks{S. Bianchi:
Universidad Nacional de Rosario, Argentina (e-mail: sbianchi@fceia.unr.edu.ar)\\
M. Escalante: Universidad Nacional de Rosario, Argentina (e-mail: mariana@fceia.unr.edu.ar)\\
G. Nasini: Corresponding author. Universidad Nacional de Rosario, Argentina
(e-mail: nasini@fceia.unr.ac.ar)\\
L. Tun\c{c}el: Department of Combinatorics and Optimization, Faculty
of Mathematics, University of Waterloo, Waterloo, Ontario N2L 3G1,
Canada (e-mail: ltuncel@uwaterloo.ca)}

\date{November 6, 2014}

\begin{abstract}
We study the Lov\'{a}sz-Schrijver lift-and-project operator ($\LS_+$) based on the cone of symmetric,
positive semidefinite matrices, applied to the fractional stable set polytope
of graphs.  The problem of obtaining a combinatorial characterization of graphs for
which the $\LS_+$-operator generates the stable set polytope in one step has been open since 1990.
We call these graphs ${\LS}_+$-\emph{perfect}.  In the current contribution, we pursue
a full combinatorial characterization of ${\LS}_+$-perfect graphs and make progress towards such
a characterization by establishing a new, close relationship among ${\LS}_+$-perfect graphs,
near-bipartite graphs and a newly introduced concept of full-support-perfect graphs.
\end{abstract}

\maketitle

\section{Introduction}

The notion of a perfect graph was introduced by Berge in the early 1960s \cite{Berge1963}.
A graph is called \emph{perfect} if each of its induced subgraphs has chromatic number equal to the
size of a maximum cardinality clique in the subgraph.
Perfect graphs have caught the attention of many researchers in the area and inspired numerous
interesting contributions to the literature for the past fifty years.
One of the main results in the seminal paper of Gr\"{o}tschel, Lov\'{a}sz and Schrijver \cite{GLS81}
is that perfect graphs constitute a graph class where the Maximum Weight Stable Set Problem (MWSSP)
can be solved in polynomial time.  Some years later, the same authors proved a very beautiful,
related result which connects a purely graph theoretic notion to polyhedrality of a
typically nonlinear convex relaxation and to the integrality and equality of two fundamental polytopes:

\newpage

\begin{theorem}
\label{thm:GLS}
(Gr\"{o}tschel, Lov\'{a}sz and Schrijver \cite{GLS81,GLS})
For every graph $G$, the following are equivalent:
\begin{enumerate}
\item
$G \textrm{ is perfect;}$
\item
$\stab(G)=\clique(G);$
\item
$\th(G)=\stab(G);$
\item
$\th(G)=\clique(G);$
\item
$\th(G) \textrm{ is polyhedral.}$
\end{enumerate}
\end{theorem}
In the above theorem, $\stab(G)$ is the stable set polytope of $G$, $\clique(G)$ is its clique relaxation
and $\th(G)$ is the \emph{theta body} of $G$ defined
by Lov\'asz \cite{Lovasz1979}.

In the early 1990s, Lov\'{a}sz and Schrijver \cite{LS91}
introduced the semidefinite relaxation ${\LS}_+(G)$ of
$\stab(G)$ which is stronger than $\th(G)$.
Following the same line of reasoning used for perfect graphs, they proved that MWSSP can be solved in
polynomial time for the class of graphs for which ${\LS}_+(G)=\stab(G)$. We call these graphs
\emph{${\LS}_+$-perfect graphs} \cite{BENT2}.  The set of ${\LS}_+$-perfect graphs is known
to contain many rich and interesting classes of graphs (e.g., perfect graphs,
$t$-perfect graphs, wheels, anti-holes, near-bipartite graphs) and their clique sums.
However, no combinatorial characterization of ${\LS}_+$-perfect graphs have been obtained so far.

There are many studies of various variants of lift-and-project operators applied
to the relaxations of the stable set problem
(see for instance, \cite{SL96,Liptak1999,deKP2002,Laurent,BO2003,LipTun2003,EMN,PVZ,GL07,GRS2013,GLRS2009}).
Why study $\LS_+$-perfect graphs?  For example, if we want to characterize the largest family
of graphs for which MWSSP can be solved in polynomial time, then perhaps, we should pick
a tractable relaxation of $\stab(G)$ that is as strong as possible.  This reasoning
would suggest that, we should focus on the strongest, tractable lift-and-project operator
and reiterate it as much as possible while maintaining tractability of the underlying
relaxation.  Even though the (lower bound) analysis for the strongest lift-and-project operators is
typically challenging, indeed, some work on the behaviour of the strongest lift-and-project
operators applied to the stable set problem already exists (see \cite{Au2014} and the references therein).
In the spectrum of strong lift-and-project operators which utilize positive semidefiniteness
constraints, given the above results of Gr\"{o}tschel, Lov\'{a}sz and Schrijver, it seems
clear to us that we should pick an operator that is at least as strong as $\th(G)$.  
Among many of the convex relaxations that are closely related to $\th(G)$ but stronger,
$\th(G)$ continues to emerge as the central object with quite special mathematical
properties (see \cite{deCarliSilva2013} and the references therein).  Given that the operator
$\LS_+(G)$ can be defined as the intersection of the matrix-space liftings of the odd-cycle polytope
of $G$ and the theta body of $G$, by definition, $\LS_+(G)$ encodes and retains very interesting combinatorial
information about the graph $G$. Then, the next question is why not focus on iterated (hence stronger)
operator $\LS_+^k$ for $k\geq 2$ but small enough to maintain tractability?  The answer to this
is related to the above; but, it is a bit more subtle: in the lifted, matrix-space representation
of $\LS_+$, if we remove the positive semidefiniteness constraint, we end up with the lifting of the operator
$\LS$ (defined later). In this lifted matrix space, if we remove the restriction that the matrix
be symmetric, we end up with the lifting of a weaker relaxation $\LS_0$.  $\LS_0^k$ retains many
interesting combinatorial properties of $G$, see \cite{Liptak1999,LiLo2001}.  Moreover, Lov\'{a}sz and Schrijver
proved that for every graph $G$, $\LS_0(G)=\LS(G)$.  However, this property  does not
generalize to the iterated operators $\LS_0^k$, $\LS^k$, even for $k=2$,
even if we require that the underlying graph $G$ be perfect (see \cite{Au2008,AuTuncel2009}).
Therefore, $\LS_+$ has many of the desired attributes for this purpose.

One of our main goals in this line of research is to obtain a characterization of ${\LS}_+$-perfect graphs
similar to the one given in Theorem \ref{thm:GLS} for perfect graphs.
More precisely, we would like to find an appropriate polyhedral relaxation
of $\stab(G)$ playing the role of
$\clique(G)$ in Theorem \ref{thm:GLS}, when we replace $\th(G)$ by ${\LS}_+(G)$.
In \cite{BENT2} we introduced the polyhedral relaxation $\NB(G)$ of $\stab(G)$,
which is, to the best of our knowledge, the tightest polyhedral relaxation of ${\LS}_+(G)$.
Roughly speaking, $\NB(G)$ is defined by the family of facets of
stable set polytopes of a family of graphs that are built from \emph{near-bipartite graphs}
by using simple operations so that the the stable set polytope of the resulting graph
does not have any facets outside the class of facets which define the stable set polytope of
near-bipartite graphs (for a precise definition of $\NB(G)$, see Section~\ref{sec:definitions}).
In our quest to obtain the desired characterization of ${\LS}_+$-perfect graphs,
$\NB(G)$ is our current best guess for replacing $\clique(G)$ in Theorem \ref{thm:GLS}.
More specifically, we conjecture that the next four statements are equivalent.

\begin{conjecture}
\label{conj0}
For every graph $G$, the following four statements are equivalent:
\begin{enumerate}
\item $\stab(G)=\NB(G)$;
\item ${\LS}_+(G)=\stab(G)$;
\item ${\LS}_+(G)=\NB(G)$;
\item ${\LS}_+(G)$ is polyhedral.
\end{enumerate}
\end{conjecture}

Verifying the validity of Conjecture~\ref{conj0}
is equivalent to determine the validity of the following two statements:

\begin{conjecture}
\label{conj1}
For every graph $G$, if
${\LS}_+(G)$ is polyhedral then $\stab(G)=\NB(G)$.
\end{conjecture}

\begin{conjecture}
\label{conj2}
For every graph $G$,
if ${\LS}_+(G)=\stab(G)$ then $\stab(G)=\NB(G)$.
\end{conjecture}

In \cite{BENT1}, we made some progress towards proving Conjecture~\ref{conj1}, by presenting
an infinite family of graphs for which it holds.  Recently, Conjecture~\ref{conj2}
was verified for web graphs \cite{EN2014}.  In this contribution, we prove
that Conjecture~\ref{conj2} holds for a class of graphs called \emph{$\fs$-perfect graphs}
that stand for \emph{full-support-perfect} graphs.
This graph family was originally defined in \cite{gra} and includes the set of
\emph{near-perfect} graphs previously defined by Shepherd in \cite{Shepherd-1994}.

One of the main difficulties in obtaining a good combinatorial characterization
for ${\LS}_+$-perfect graphs is that the lift-and-project operator ${\LS}_+$
(and many related operators) can behave sporadically under many well-studied graph-minor operations
(see \cite{EMN,LipTun2003}).  Therefore, in the study of ${\LS}_+$-perfect graphs
we are faced with the problem of constructing suitable graph operations and then
deriving certain monotonicity or loose invariance properties under such
graph operations. In this context, we present two operations which
\emph{preserve} ${\LS}_+$-\emph{imperfection} in graphs.

In the next section, we begin with notation and preliminary results that will be used
throughout the paper. We also state our main characterization
conjecture on ${\LS}_+$-perfect graphs. In Section \ref{Fk},
we characterize $\fs$-perfection in the family of graphs built from a minimally imperfect graph
and one additional node.
In Section \ref{thevalidity} we prove the validity of the conjecture
on $\fs$-perfect graphs.  In order to ease the reading of this contribution,
the proofs of results on the ${\LS}_+$ operator are presented in Section \ref{lodelN+}.
Section \ref{conclu} is devoted to the conclusions and some further results.

\section{Further definitions and preliminary results} \label{sec:definitions}

\subsection{Graphs and the stable set polytope}

Throughout this work, $G$ stands for a simple graph with node set $V(G)$ and edge set $E(G)$. The \emph{complementary graph} of $G$, denoted by $\overline{G}$, is such that $V(\overline{G})=V(G)$ and, for $E(\overline{G})=\{uv: u,v\in V(G), uv \notin E(G)\}$.  For any positive integer $n$, $K_n$, $C_n$ and $P_n$ denote the graphs with $n$ nodes corresponding to a  \emph{complete graph}, a \emph{cycle} and a \emph{path}, respectively. We assume that in the cycle $C_{n}$ node $i$ is adjacent to node $i+1$ for $i \in \{1,\ldots, n-1\}$ and $n$ is adjacent to node $1$.

Given $V'\subseteq  V(G)$, we say that $G'$ is a subgraph of $G$ \emph{induced by the nodes in} $V'$ if $V(G')=V'$ and $E(G')=\{uv: uv\in E(G), \{u,v\}\subseteq  V(G')\}$. When $V(G')$ is clear from the context we say that $G'$ is a node induced subgraph of $G$ and write $G'\subseteq G$.  Given $U\subseteq  V(G)$, we denote by $G-U$ the subgraph of $G$ induced by the nodes in $V(G)\setminus U$.
For simplicity, we write $G-u$ instead of $G-\{u\}$.  We say that $G_E$ is an \emph{edge subgraph} of $G$ if $V(G_E)=V(G)$ and $E(G_E)\subseteq E(G)$.

Given the graph $G$, the set $\Gamma_G(v)$ is the \emph{neighbourhood} of node $v\in V(G)$ and $\delta_G(v)=|\Gamma_G(v)|$.
The set $\Gamma_G[v]=\Gamma_G(v)\cup \{v\}$ is the \emph{closed neighbourhood} of node $v$. When the graph is clear from context we simply write $\Gamma(v)$ or $\Gamma[v]$.
If $G'\subseteq G$ and $v\in V(G)$, $G'\ominus v$ is the subgraph of $G$ induced by the nodes in $V(G')\setminus \Gamma[v]$. We say that $G'\ominus v$ is obtained from $G'$ by \emph{destruction} of $v\in V(G)$.

A \emph{stable set} in $G$ is a subset of mutually nonadjacent nodes in $G$ and a \emph{clique} is a subset of pairwise adjacent nodes in $G$. The cardinality of a stable set of maximum cardinality in $G$ is denoted by $\alpha(G)$. The stable set polytope in $G$, $\stab(G)$, is the convex hull of the characteristic vectors of stable sets in $G$. The \emph{support} of a valid inequality of the stable set polytope of a graph $G$ is the subgraph induced by the nodes with nonzero coefficient in the inequality and a \emph{full-support} inequality has $G$ as support.

If $G'\subseteq  G$ we can consider every point in
$\stab(G')$ as a point in $\stab(G)$, although they do not belong
to the same space (for the missing nodes, we take direct sums with the
interval $[0,1]$, since originally
$\stab(G) \subseteq  \stab(G') \oplus [0,1]^{V(G) \setminus V(G')}$).
Then, given any family of graphs $\mathcal{F}$ and a graph $G$,
we denote by $\mathcal{F}(G)$
the relaxation of $\stab(G)$ defined by
\begin{equation}\label{F}
\mathcal{F}(G)=\bigcap\limits_{G'\subseteq  G; G'\in \mathcal{F}} \stab(G').
\end{equation}

If $\fra$ denotes the family of complete graphs of size two, following the definition (\ref{F}), the polyhedron $\fra(G)$ is called the \emph{edge relaxation}. It is known that  $G$ is bipartite if and only if $\stab(G)= \fra(G)$.
Similarly, if $\clique$ denotes the family of complete graphs,
$\clique(G)$ is the \emph{clique relaxation} already mentioned
and a graph is perfect if and only if $\stab(G)=\clique(G)$ \cite{chvatal}.
Moreover, if $\oddcycle$ denotes the family of odd cycles, as a consequence of results in \cite{LS91} we have the following
\begin{remark}\label{Gmenosv}
If $G-v$ is bipartite for some $v\in V(G)$ then $\stab(G)= \fra(G) \cap \oddcycle(G)$.
\end{remark}

In \cite{Shepherd-1995} Shepherd defined a graph $G$ to be near-bipartite if $G\ominus v$ is bipartite for every $v\in V(G)$.
We denote by $\NB$ the family of near-bipartite graphs.
Since complete graphs and odd cycles are near-bipartite graphs, it is clear that
\[\NB(G)\subseteq   \clique(G)\cap \oddcycle(G).\]

\subsection{Minimally imperfect, near-perfect and fs-perfect graphs}

\emph{Minimally imperfect} graphs are those graphs that are not perfect but after deleting any node they become perfect.
The Strong Perfect Graph Theorem  \cite{CRST2006} (also see \cite{CRST2003}; and see \cite{Cor} for the related
recognition problem) states that the only minimally imperfect graphs
are the odd cycles and their complements.

Given a graph $G$ it is known that the \emph{full-rank constraint}
\[\sum_{u\in V} x_u\leq\alpha(G)\]
is always valid for $\stab(G)$.
A graph is \emph{near-perfect} if its stable
set polytope is defined only by non-negativity constraints, clique constraints
and the full-rank constraint \cite{Shepherd-1994}.
Due to the results of Chv\'atal \cite{chvatal}, near-perfect graphs define a superclass of perfect graphs and after \cite{Pad} minimally imperfect graphs are also near-perfect.
Moreover, every node induced subgraph of a near-perfect graph is near-perfect \cite{Shepherd-1994}.
In addition, Shepherd \cite{Shepherd-1994}
conjectured that near-perfect graphs could be characterized in terms of certain combinatorial parameters
and established that the validity of the conjecture follows from the Strong Perfect Graph (then Conjecture, now Theorem). Therefore,

\begin{theorem}(\cite{CRST2006},\cite{Shepherd-1994})\label{np}
A graph $G$ is near-perfect if and only if, for every $G'\subseteq  G$ minimally imperfect, the following two properties hold:
\begin{enumerate}
\item $\alpha(G')=\alpha(G)$;
\item for all $v\in V(G)$, $\alpha(G'\ominus v)=\alpha(G)-1$.
\end{enumerate}
\end{theorem}

As a generalization of near-perfect graphs we consider the family of \emph{$\fs$-perfect} (full-support perfect) graphs. A graph is $\fs$-perfect if its stable set polytope is defined only by non-negativity constraints, clique constraints and at most one single full-support inequality.
Then, every node induced subgraph of an $\fs$-perfect graph is $\fs$-perfect.
We say that a graph is \emph{properly} $\fs$-perfect if it is an imperfect $\fs$-perfect graph.
Clearly, near-perfect graphs are $\fs$-perfect but we will see that  $\fs$-perfect graphs define a strict superclass of near-perfect graphs.

\subsection{The ${\LS}_+$ operator}

In this section, we present the definition of the ${\LS}_+$-operator \cite{LS91} and some of its well-known properties when it is applied to relaxations of the stable set polytope of a graph.
In order to do so, we need some more notation and definitions.

We denote by $\e_0, \e_1, \dots, \e_n$ the vectors of the canonical basis of $\R^{n+1}$ where the first coordinate is indexed zero.
Given a convex set $K$ in $[0,1]^n$,
 \[
 \cone(K) := \left\{\left(\begin{array}{c} x_0\\ x \end{array}
 \right) \in \R^{n+1}: x=x_0 y ; y\in K; x_0 \geq 0 \right\}.
 \]
Let $\S^{n}$ be the space of $n\times n$ symmetric matrices with real entries. If $Y\in \S^n$, $\diag (Y)$ denotes the vector whose $i$-th entry is $Y_{ii}$, for every $i \in \{1,\ldots,n\}$. Let

 \begin{eqnarray*}
 M(K) := \left\{Y \in \S^{n+1}: \right.
 & & Y\e_0 = \diag (Y),\\
 & & Y\e_i \in \cone(K), \\
 & & \left. Y (\e_0 - \e_i) \in \cone(K),
 \forall i \in \{1,\dots,n \} \right\}.
 \end{eqnarray*}

Projecting this polyhedral lifting back to the space $\R^n$  results in
 \[
 \LS(K) := \set{ x \in [0,1]^n : \left(\begin{array}{c} 1\\ x \end{array}
 \right)= Y \e_0, \mbox{ for some } Y \in M(K)}.
 \]
Clearly, $\LS(K)$ is a relaxation of the convex hull of integer solutions in $K$, i.e., $\conv(K\cap \{0,1\}^n)$.

Let $\S_+^{n}$ be the space of $n\times n$ symmetric
positive semidefinite (PSD) matrices with real entries.  Then
\[
M_+(K) := M(K) \cap \S_+^{n+1}
\]
yields the tighter relaxation
\[
 {\LS}_+(K) := \set{ x \in [0,1]^{n} : \left(\begin{array}{c} 1\\ x \end{array}
 \right)= Y \e_0, \mbox{ for some } Y \in M_+(K)}.
 \]

If we let $\LS^0(K):=K$, then the successive applications of the $\LS$ operator yield $\LS^k(K)=\LS(\LS^{k-1}(K))$ for every $k\geq 1$. Similarly for the ${\LS}_+$ operator.  Lov\'{a}sz and Schrijver proved that $\LS^n(K)={\LS}_+^n(K)=\conv(K\cap \{0,1\}^n)$.

In this paper, we focus on the behaviour of the  ${\LS}_+$ operator on the edge relaxation of the stable set polytope.
In order to simplify the notation we write ${\LS}_+(G)$ instead of ${\LS}_+(\fra(G))$ and similarly for the successive iterations of it.
It is known \cite{LS91} that, for every graph $G$,
\[
\stab(G) \subseteq   {\LS}_+(G)\subseteq   \th(G)\subseteq   \clique(G)
\]
and
\[
\stab(G) \subseteq   {\LS}_+(G)\subseteq   \NB(G).
\]

\subsection{${\LS}_+$-perfect graphs}
Recall that a graph $G$ is ${\LS}_+$-perfect if ${\LS}_+(G) = \stab(G)$. A graph that is not ${\LS}_+$-perfect is called \emph{${\LS}_+$-imperfect}.

Using the results in \cite{EMN} and \cite{LipTun2003}
we know that all imperfect graphs with at most $6$ nodes
are ${\LS}_+$-perfect,
except for the two properly near-perfect graphs depicted in Figure \ref{grafos6}, denoted by $G_{LT}$ and $G_{EMN}$, respectively.
These graphs prominently figure into our current work as the building blocks of an interesting family of graphs.

\begin{figure}[h]
\begin{center}
\includegraphics{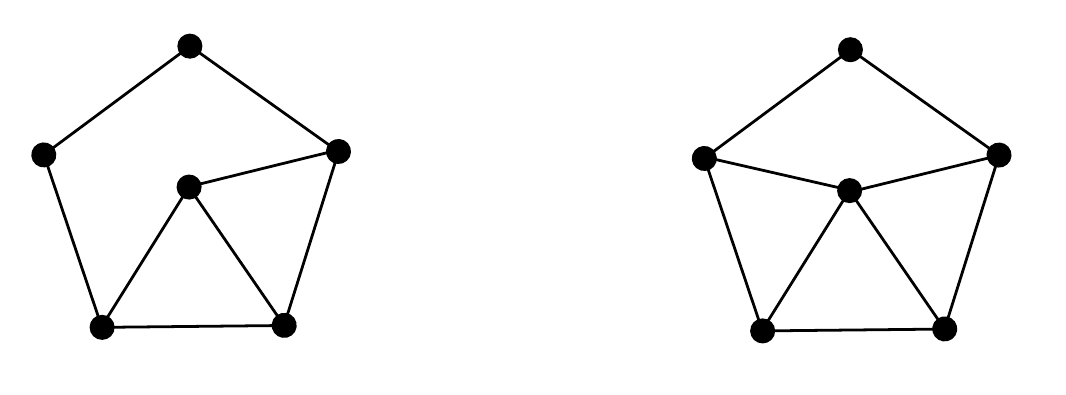}
\caption{The graphs $G_{LT}$ and $G_{EMN}$.}
\label{grafos6}
\end{center}
\end{figure}

From the results in \cite{LS91}, it can be proved that every subgraph of an ${\LS}_+$-perfect graph is also ${\LS}_+$-perfect. Moreover, every graph for which $\stab(G)=\NB(G)$ is ${\LS}_+$-perfect. In particular, perfect and near-bipartite graphs are ${\LS}_+$-perfect.  Recall that in Conjecture~\ref{conj1} we wonder whether the only ${\LS}_+$-perfect graphs are those graphs $G$ for which $\stab(G)=\NB(G)$.
Obviously, $G$ is ${\LS}_+$-perfect if and only if every facet defining inequality of $\stab(G)$ is valid for ${\LS}_+(G)$. In this context, we have the Lemma 1.5 in \cite{LS91} that can be rewritten in the following way:

\begin{theorem}
Let $ax\leq \beta$ be a full-support valid inequality for $\stab(G)$. If, for every $v\in V(G)$, $\sum_{w\in V(G-v)} ax\leq \beta- a_v$ is valid for $\fra(G\ominus v)$ then $ax\leq \beta$ is valid for ${\LS}_+(G)$.
\end{theorem}

In \cite{BENT1} we proved that the converse of the previous result is not true. However, it is plausible that the converse holds when the full-support valid inequality is a facet defining inequality of $\stab(G)$. Actually, the latter assertion would be a consequence of the validity of Conjecture~\ref{conj1}. Thus, we present an equivalent formulation of it in the following.

\begin{conjecture} \label{conjecture}
If a graph is ${\LS}_+$-perfect and its stable set polytope has a full-support facet defining inequality, then the graph is near-bipartite.
\end{conjecture}

\subsection{Graph operations}\label{operation}

In this section, we present some properties of four graphs operations that will be used throughout this paper.
Firstly, let us recall the \emph{complete join} of graphs.
Given two graphs $G_1$ and $G_2$ such that $V(G_1)\cap V(G_2)=\emptyset$, we say that a graph $G$ is obtained after the complete join of $G_1$ and $G_2$, denoted $G=G_1 \vee G_2$, if $V(G)=V(G_1) \cup V(G_2)$ and $E(G)= E(G_1) \cup E(G_2) \cup \{vw: v\in V(G_1)\text{ and }w\in V(G_2)\}$.
A simple example of a join is the $n$-\emph{wheel} $W_n$, for $n\geq 2$  which is the complete join of the trivial graph with one node and the $n$-cycle.

It is known that every facet defining inequality of $\STAB(G_1 \vee G_2)$ can be obtained by the cartesian product of facets of $\STAB(G_1)$ and $\STAB(G_2)$.
Hence, odd wheels are $\fs$-perfect but not near-perfect graphs. Moreover, it is easy to see
\begin{remark}\label{joinclique}
The complete join of two graphs is properly $\fs$-perfect if and only if one of them is a complete graph and the other one is a properly $\fs$-perfect graph.
\end{remark}

Also, it is known that
\begin{remark}\label{necesaria}
The complete join of two graphs is ${\LS}_+$-perfect if and only if both of them are ${\LS}_+$-perfect graphs.
\end{remark}

Now, let us recall the \emph{odd-subdivision of an edge} \cite{Wol76}.
Given a graph $G=(V,E)$ and $e\in E$, we say that the graph $G'$ is obtained from $G$ after the odd-subdivision of the edge $e$ if it is replaced in $G$ by a path of odd length.
Next, we consider the $k$-\emph{stretching of a node} which is a generalization of the \emph{type (i) stretching operation} defined in \cite{LipTun2003}.
Let $v$ be a node of $G$ with neighborhood $\Gamma(v)$ and let $A_1$ and $A_2$ be nonempty
subsets of $\Gamma(v)$ such that $A_1 \cup A_2 = \Gamma(v)$, and $A_1\cap A_2$ is a clique of size $k$.
A $k$-stretching of the node
$v$ is obtained as follows:
remove $v$, introduce three nodes instead, called $v_1$, $v_2$ and $u$, and add an edge
between
$v_i$ and every node in $\{u\} \cup A_i$ for $i \in \{1, 2\}$. Figure \ref{1stretching} illustrates the case $k=1$.

\begin{figure}[h]
\begin{center}
\includegraphics{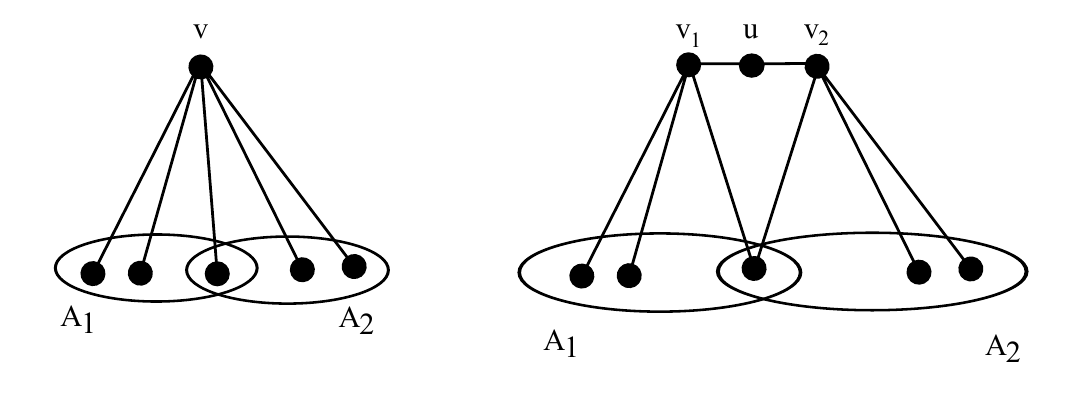}
\caption{A $1$-stretching operation on node $v$.}
\label{1stretching}
\end{center}
\end{figure}

The type (i) stretching operation presented in \cite{LipTun2003} corresponds to the case $k=0$.

We also consider another graph operation defined in \cite{AEF}.
Given the graph $G$ with nodes $\{1, \dots,n\}$ and a clique $K = \{v_1,\dots, v_s\}$ in $G$ (not necessarily maximal),  the \emph{clique subdivision of the edge} $v_1 v_2$ \emph{in } $K$ is defined as follows: delete the edge $v_1 v_2$ from $G$, add the nodes $v_{n+1}$ and $v_{n+2}$ together
with the edges $v_1 v_{n+1}$, $v_{n+1} v_{n+2}$, $v_{n+2} v_2$ and $v_{n+i} v_j$ for $i \in \{1, 2\}$ and $j \in  \{3,\dots, s\}$.
Figure \ref{clique_subdivision} illustrates the clique subdivision of the edge $v_1 v_2$
in the clique $K = \{v_1, v_2, v_3\}$.
\begin{figure}
\centering
\includegraphics{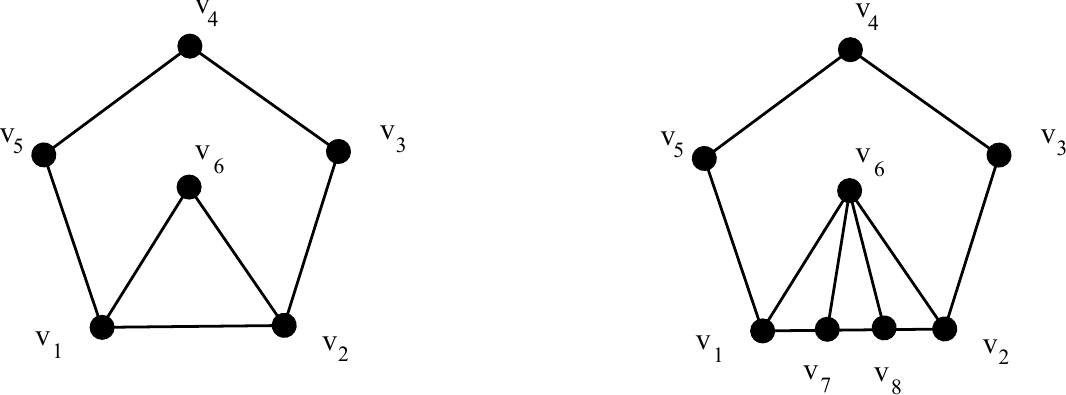}
\caption{The graph $G'$ is obtained from $G$ after the clique subdivision of edge $v_1v_2$.}
\label{clique_subdivision}
\end{figure}
Notice that if the clique is an edge this operation reduces to the
odd-subdivision of it.

\section{$\fs$-perfection on graphs in $\F^k$}\label{Fk}

In order to prove Conjecture \ref{conjecture} on $\fs$-perfect graphs, we first consider a minimal structure that a graph must have in order to be properly $\fs$-perfect and ${\LS}_+$-imperfect.
This leads us to define $\F^k$ for every $k\geq 2$ as the family of graphs having node set  $\{0,1,\dots,2k+1\}$ and such that $G-0$ is a minimally imperfect graph with $1\leq \delta_G(0)\leq 2k$.
Let us consider necessary conditions for a graph in $\F^k$ to be $\fs$-perfect.

\begin{theorem} \cite{gra} \label{alfa2}
Let $G\in \F^k$ be an $\fs$-perfect graph. Then, the following conditions hold:
\begin{enumerate}
\item $\alpha(G)=\alpha(G-0)$;
\item $1\leq \alpha(G\ominus 0) \leq \alpha(G)-1$;
\item the full-support facet defining inequality of $\stab(G)$ is the inequality
\begin{equation*}
\left(\alpha(G)-\alpha(G \ominus 0)\right)x_0+\sum_{i=1}^{2k+1} x_i\leq \alpha(G).
\end{equation*}
\end{enumerate}
\end{theorem}

\begin{proof}
Let
\begin{equation}\label{full}
\sum_{i=0}^{2k+1}a_i x_i\leq \beta
\end{equation}
be the full-support facet defining inequality of $\STAB(G)$. We may assume that all coefficients $a_i$,  $i \in \{0,\dots,2k+1\}$ are positive integers.
Clearly,
\[
\STAB(G-0)=\left\{x\in \QSTAB(G-0): \sum_{i=1}^{2k+1} a_i x_i\leq \beta\right\}.
\]
Since $G-0$ is a minimally imperfect graph, the inequality $\sum_{i=1}^{2k+1} a_i x_i\leq \beta$ is a
positive multiple of its rank constraint, i.e., there exists a positive integer number $p$ such that
$a_i= p$ for $i \in \{1,\dots, 2k+1\}$ and $\beta= p \ \alpha(G-0).$
Therefore, \eqref{full} has the form
\begin{equation} \label{facet}
a_0 x_0+ p\sum_{i=1}^{2k+1} x_i\leq p \ \alpha(G-0).
\end{equation}
Observe that there is at least one root $\tilde{x}$ of \eqref{facet} such that $\tilde{x}_0=1$.
Clearly, $\tilde{x}$ is the incidence vector of a stable set $S$ of $G$ such that $S-\{0\}$ is a maximum
stable set of $G\ominus 0$. Then,
\[
a_0= p \ \left(\alpha(G-0)- \alpha(G\ominus 0)\right).\]

Since $a_0\geq 1$, we have
$\alpha(G\ominus 0)\leq \alpha(G-0)-1$. Moreover, since $\delta_G(0)\leq 2k$, then $\alpha(G\ominus 0)\geq 1$.
Therefore, the inequality \eqref{facet} becomes
\begin{equation} \label{facet2}
(\alpha(G-0)- \alpha(G\ominus 0)) x_0+ \sum_{i=1}^{2k+1} x_i\leq \alpha(G-0).
\end{equation}
To complete the proof, we only need to show that $\alpha(G)=\alpha(G-0)$.
Let $\bar{x}$ be the incidence vector of a maximum stable set in $G$, then
\begin{equation}\label{stableG}
\alpha(G)=\bar{x}_0+\sum_{i=1}^{2k+1}\bar{x}_i.
\end{equation}
Moreover, since $\alpha(G-0)- \alpha(G\ominus 0)\geq 1$ and $\bar{x}$ satisfies  \eqref{facet2} we have
\[
\alpha(G)=\bar{x}_0+\sum_{i=1}^{2k+1}\bar{x}_i\leq (\alpha(G-0)- \alpha(G\ominus 0)) \bar{x}_0+ \sum_{i=1}^{2k+1} \bar{x}_i\leq \alpha(G-0).
\]
We have that $\alpha(G)\leq \alpha(G-0)$, implying $\alpha(G)=\alpha(G-0)$.
\end{proof}

As a first consequence of the previous theorem we have:

\begin{corollary} \label{cor_alfa21}
Let  $G\in \F^k$ be such that $\alpha(G)=2$. Then, $G$ is $\fs$-perfect if and only if $G-0=\overline{C_{2k+1}}$ (the complementary graph of $C_{2k+1}$) and $G$ is near-perfect.
\end{corollary}

\begin{proof}
Assume that $G$ is $\fs$-perfect. By the previous theorem, $\alpha(G-0)= \alpha(G)=2$ and then $G-0=\overline{C_{2k+1}}$. Moreover, $1\leq \alpha(G\ominus 0)\leq \alpha(G)-1=1$ and $a_0=1$. Thus $G$ is near-perfect.
The converse follows from the definition of $\fs$-perfect graphs.
\end{proof}

For $k\geq 2$, let $H^k$ denote the graph in $\F^k$ having $\alpha(H^k)=2$ and $\delta_{H^k}(0)=2k$. Using Theorem \ref{np} it is easy to check that $H^k$ is near-perfect.
Using Corollary \ref{cor_alfa21}, we have the following result:

\begin{corollary} \label{cor_alfa22}
Let $G\in \F^k$ be such that $\alpha(G)=2$. Then, $G$ is $\fs$-perfect if and only if $G$ is a near-perfect edge subgraph of $H^k$.
\end{corollary}

Let us now study the structure of $\fs$-perfect graphs $G$ in $\F^k$ with stability number at least 3. From Theorem \ref{alfa2}, $\alpha(G-0)=\alpha(G)\geq 3$ and  $G-0=C_{2k+1}$ with $k\geq 3$.
Recall that in the cycle $C_{2k+1}$ node $i$ is adjacent to node $i+1$ for $i \in \{1,\ldots, 2k\}$ and $2k+1$ is adjacent to node $1$.
Clearly, if $\delta(0)\leq 2$ and $0v \in E(G)$ then $G-v$ is bipartite and, by Remark \ref{Gmenosv}, $G$ is not $\fs$-perfect.

From now on, $3\leq s=\delta(0)\leq 2k$ and $\Gamma(0)=\{v_1,\ldots,v_s\}$ such that $1\leq v_1<\ldots <v_{s}\leq 2k+1$. Observe that for every $i \in \{1,\ldots,s-1\}$, the nodes in $\{w\in V(G-0): v_i\leq w \leq v_{i+1}\}$  together with node $0$ form a chordless cycle $D_i$ in $G$.  Also, $D_s$ in $G$ is the chordless cycle induced by the nodes in $\{w\in V(G-0): v_s\leq w \leq 2k+1 \text{ or } 1\leq w \leq v_1 \}$ and node $0$. We refer to these cycles as \emph{central cycles of} $G$.  It is easy to see, using parity arguments, that every $G\in \F^k$ has and odd number of odd central cycles.
If $G$ has only one odd central cycle, say $D_1$, then $G-v_1$ is bipartite and, by Remark \ref{Gmenosv}, $G$ is not $\fs$-perfect.

We summarize the previous ideas in the following result:

\begin{lemma} \label{nofs}
Let $G\in \F^k$ be a $\fs$-perfect graph with $\alpha(G)\geq 3$. Then, $k\geq 3$, $G-0=C_{2k+1}$ and $G$ has at least three odd central cycles.
\end{lemma}

According to lemma above, we can focus on structural properties of graphs in $\F^k$ with at least $3$ odd central cycles.
Firstly, we have:

\begin{lemma}\label{oddsubdiv}
Let $G\in \F^k$ with $k\geq 3$ be such that $G-0=C_{2k+1}$ and $G$ has at least $3$ odd central cycles.  Then, $G$ can be obtained after odd subdivisions of edges in a graph $G'\in \F^p$ for some $2\leq  p<k$ with $\delta_{G'}(0)=\delta_{G}(0)$. Moreover,
\begin{enumerate}
	\item if every central cycle of $G$ is odd, $G'$ has one central cycle of length $5$ and $2(p-1)$ central cycles of length $3$;
	\item if $G$ has an even central cycle, every central cycle of $G'$ has length $3$ or $4$.
	\end{enumerate}
\end{lemma}

\begin{proof}
Let $\delta_{G}(0)=s$ with $s\geq 3$.
\begin{enumerate}
	\item If every central cycle of $G$ is odd then $s$ is odd and $G$ has a central cycle with length at least $5$. Let $p=\frac{s+1}{2}$ and $G'\in \F^p$ with $\delta_{G'}(0)=s$ such that $0$ is adjacent to all nodes in $C_{2p+1}$ except to two nodes, e.g., nodes $s$ and $s+1$. Clearly, $G'$ has one $5$-central cycle and $s-1=2(p-1)$ central cycles of length $3$. Hence $G$ is obtained after the odd subdivision of edges of $G'$.
	\item If $G$ has $r\geq 1$ even central cycles then $s+r$ is odd. Let $D_i$ with $i \in \{1,\ldots, s \}$ the central cycles of $G$. Let $p= \frac{s+r-1}{2}$ and $G'\in \F^p$ such that $\delta_{G'}(0)=s$ and
for $i \in \{1,\ldots, s\}$ the central cycle $D'_i$ in $G'$ is:
\begin{enumerate}
	\item a $3$-cycle if $D_i$ is odd,
	\item a $4$-cycle If $D_i$ is even.
\end{enumerate}
It is straightforward to check that $G$ is obtained from $G'$ after odd subdivision of edges.
\end{enumerate}
\end{proof}

In addition, we have

\begin{lemma}\label{operaciones}
Let $G\in \F^k$ with $k\geq 3$ be such that $G-0=C_{2k+1}$ and $\delta_G(0)\geq 3$. Let $t(G)$ be the number of $3$-central cycles and $r(G)$ the number of $4$-central cycles in $G$.
\begin{enumerate}
\item if $G$ has three consecutive $3$-central cycles then $G$ can be obtained after the clique subdivision of an edge in a graph $G'\in \F^{k-1}$ with $t(G')=t(G)-2$ and $\delta_{G'}(0)=\delta_G(0)-2$;
\item if $r(G)\geq 2$ then $G$ can be obtained after the 1-stretching operation on a node in a graph $G'\in \F^{k-1}$ with $r(G')=r(G)-1$ and $\delta_{G'}(0)=\delta_{G}(0)-1$.
\end{enumerate}
\end{lemma}

\begin{proof}
\begin{enumerate}
\item From assumption we may consider that $\{2k-1,2k,2k+1,1\}\subseteq  \Gamma(0)$.
Let $G'\in \F^{k-1}$ be a graph having $\Gamma_{G'}(0)=\Gamma(0)\setminus \{2k,2k+1\}$.
Clearly, $G'$ has $t(G)-2$ $3$-central cycles and $\delta_{G'}(0)=\delta_{G}(0)-2$. Moreover, it is easy to see that $G$ is the clique subdivision of the edge in $G'$ having extreme points $\{1,2k-1\}$.
\item Since there is a $4$-central cycle, without loss of generality, we may assume that the nodes in $\{0,1,2k, 2k+1\}$ induce a 4-central cycle in $G$.
Consider $G'\in F^{k-1}$ be such that $\Gamma_{G'}(0)=\Gamma_G(0)\setminus \{2k\}$ then $G$ is a 1-stretching of $G'$ performed at node 1 and where the new nodes are $2k$ and $2k+1$. Clearly, $r(G')=r(G)-1$ and $\delta_{G'}(0)=\delta_{G}(0)-1$.
\end{enumerate}
\end{proof}

Utilizing the previous lemmas we obtain the following result.

\begin{theorem}\label{reduccion}
Let $G\in \F^k$ with $k\geq 3$ and $G$ has at least three odd central cycles. Then, $G$ can be obtained from $G_{LT}$ or $G_{EMN}$ after successively applying odd-subdivision of an edge, 1-stretching of a node and clique-subdivision of an edge operations.
\end{theorem}

\begin{proof}
Using Lemma \ref{oddsubdiv}, by successively  performing the odd subdivision of an edge operation, we can restrict ourselves to consider the following cases:
\begin{enumerate}
	\item[(a)] $G$ has one $5$-central cycle and $2(k-1)$ $3$-central cycles.
	\item[(b)] $G$ has at least one even central cycle and every central cycle has length $3$ or $4$.
\end{enumerate}
Consider $r(G)$, the number of $4$-cycles in $G$.
If $r(G)=0$ then $G$ is a graph described in case (a). Since $k\geq 3$ then $2(k-1)\geq 4$ and by Lemma \ref{operaciones} (1), we can conclude that $G$ is obtained from $G_{LT}$ after successive clique-subdivisions of edges.
For graphs in case (b) we have that $r(G)\geq 1$ and  then by Lemma \ref{oddsubdiv} (2) we have that $2k=r(G)+\delta(G)-1$.
If $r(G)=1$ and $\delta(G)$ is even, since $k\geq 3$ it holds that $G$ has $t(G)=\delta(G)-1\geq 5$ number of $3$-cycles. Using Lemma \ref{operaciones} (1) it is not hard to see that $G$ can be obtained from $G_{EMN}$ by successive clique subdivisions of edges.
If $r(G)\geq 2$, Lemma \ref{operaciones} (2) implies that $G$ can be obtained by successive 1-stretching of nodes from a graph $G'\in \F^{k'}$ with $r(G')=1$ and $2k'=\delta(G')$. If $2k'=4$ then $G'=G_{EMN}$ otherwise we can refer to the previous case and the proof is complete.
\end{proof}

\section{The conjecture on $\fs$-perfect graphs} \label{thevalidity}

In this section, we prove the validity of Conjecture \ref{conjecture} on the family of $\fs$-perfect graphs.
We start by proving it on graphs in the family $\F^k$ with $k\geq 2$.
Recall that when $G\in \F^k$ is a $\fs$-perfect graph with $\alpha(G)=2$, Corollary \ref{cor_alfa22} states that $G$ is a near-perfect edge subgraph of $H^k$.
Observe that $H^2=G_{EMN}$ and then, due to the results in \cite{EMN}, $H^k$ is ${\LS}_+$-imperfect when $k=2$.
Next, we prove the imperfection property for the whole family of graphs $H^k$.

\begin{theorem}\label{oddanti}
For $k\geq 2$, the graph $H^k$ is ${\LS}_+$-imperfect.
\end{theorem}

In order to ease the reading of this paper we postpone the proof of Theorem \ref{oddanti} to Section \ref{lodelN+}.
On the behaviour of the ${\LS}_+$ operator on edge subgraphs, we have the following result:

\begin{lemma}
Let $G_E$ be an edge subgraph of $G$ and $a x \leq \beta$ be a valid inequality for $\stab(G_E)$. Then, if  $a x \leq \beta$ is not valid for $\LS^r_+(G)$, then $\stab(G_E)\neq \LS^r_+(G_E)$.
\end{lemma}

\begin{proof}
Clearly, by definition $\LS^r_+(G)\subseteq \LS^r_+(G_E)$.
Thus, if there exists $\hat{x}\in \LS^r_+(G)$ such that $a \hat{x} > \beta$ then $a x\leq \beta$ is not valid for $\LS^r_+(G_E)$. Moreover, by hypothesis, $a x\leq \beta$ is valid for $\stab(G_E)$ and the result follows.
\end{proof}

As a consequence of the above, we have:

\begin{theorem} \label{antiholes}
Let $G \in \mathcal{F}^k$ be a $\fs$-perfect graph with $\alpha(G)=2$. Then, $G$ is ${\LS}_+$-imperfect.
\end{theorem}

\begin{proof}
  By Corollary \ref{cor_alfa22} we know that $G$ is a near-perfect edge subgraph of $H^k$. Theorem \ref{oddanti} states that $H^k$ is ${\LS}_+$-imperfect then the full rank constraint is not valid for ${\LS}_+(H^k)$. Since $\alpha(H^k)=\alpha(G)=2$, the full rank constraint is not valid for ${\LS}_+(G)$ and $G$ is ${\LS}_+$-imperfect.
\end{proof}

Let us consider the $\fs$-perfect graphs $G$ in $\F^k$ with $\alpha(G)\geq 3$.
Due to the structural characterization in Theorem \ref{reduccion}, we are interested in the behavior of the ${\LS}_+$-operator
under the odd subdivision of an edge, $k$-stretching of a node and clique subdivision of an edge operations.
In this context, a related earlier result is:

\begin{theorem}[\cite{LipTun2003}]\label{}
Let $G$ be a graph and $r \geq 1$ such that ${\LS}_+^r(G) \neq \stab(G)$.  Further assume
that $\tilde{G}$ is obtained from $G$ by using the odd subdivision
operation on one of its edges.  Then, ${\LS}_+^r(\tilde{G}) \neq \stab(\tilde{G})$.
\end{theorem}

Concerning the remaining operations, we present the following results whose proofs are included in Section \ref{lodelN+} for the sake of clarity.

\begin{theorem}\label{Thm2.10}
Let $G$ be a graph and $r \geq 1$ such that ${\LS}_+^r(G) \neq \stab(G)$.  Further assume
that $\tilde{G}$ is obtained from $G$ by using the $k$-stretching
operation on one of its nodes.  Then, ${\LS}_+^r(\tilde{G}) \neq \stab(\tilde{G})$.
\end{theorem}

\begin{theorem}\label{cliquesubdivision}
Let $G$ be a graph and $r \geq 1$ such that ${\LS}_+^r(G) \neq \stab(G)$.  Further assume
that $\tilde{G}$ is obtained from $G$ by using the clique subdivision
operation on one of its edges.  Then, ${\LS}_+^r(\tilde{G}) \neq \stab(\tilde{G})$.
\end{theorem}

In summary, we can conclude that the odd-subdivision of an edge, the $1$-stretching of a node and the clique-subdivision of an edge are operations that \emph{preserve} ${\LS}_+$-imperfection.
Then, the behavior of these operations under the ${\LS}_+$ operator together with the fact that graphs $G_{LT}$ and $G_{EMN}$ are ${\LS}_+$-imperfect, Lemma \ref{nofs} and Theorem \ref{reduccion} allow us to deduce:

\begin{theorem}\label{join0}
Let $G\in \F^k$ be a $\fs$-perfect graph with $\alpha(G)\geq 3$. Then, $G$ is ${\LS}_+$-imperfect.
\end{theorem}

Finally, we are able to present the main result of this contribution.

\begin{theorem}\label{main}
Let $G$ be a properly $\fs$-perfect graph which is also ${\LS}_+$-perfect. Then, $G$ is the complete join of a complete graph (possibly empty) and a minimally imperfect graph.
\end{theorem}

\begin{proof}
Since $G$ is a properly $\fs$-perfect graph, $G$ has a $(2k+1)$-minimally imperfect node induced subgraph $G'$. If $G'=G$ the theorem follows. Otherwise, let $v\in V(G)\setminus V(G')$ and let $G_v$ be the subgraph of $G$ induced by $\{v\}\cup V(G')$. Clearly, $G_v$ is properly $\fs$-perfect as well as ${\LS}_+$-perfect.
Then, by Theorem \ref{antiholes} and Theorem \ref{join0}, $G_v\notin \F^k$. So, $\delta_{G_v}(v)=2k+1$ and $G_v=\{v\}\vee G'$.
Therefore, if $G''$ is the subgraph of $G$ induced by $V(G)-V(G')$, $G=G' \vee G''$. By Remark \ref{joinclique}, $G''$ is a complete graph and by Remark \ref {necesaria} the result follows.
\end{proof}

Since complete joins of  complete graphs and minimally imperfect graphs are near-bipartite,  they satisfy $\nb(G)=\STAB(G)$.  Therefore,
based on the results obtained so far, we can conclude that Conjecture \ref{conjecture} holds for $\fs$-perfect graphs.

\section{Results concerning the ${\LS}_+$-operator}
\label{lodelN+}

In this section we include the proofs of some results on the ${\LS}_+$-operator that were stated without proof in the previous sections.

\subsection{The ${\LS}_+$-imperfection of the graph $H^k$.}

Recall that $V(H^k)=\{0,1,\ldots,2k+1\}$, $H^k-0=\overline{C_{2k+1}}$ and $\delta(0)=2k$.
Without loss of generality, we may assume that the node $2k+1$ in $H^k$ is the only one not connected with node $0$.
Let us introduce the point $x(k,\gamma)=\frac{1}{2k+2+\gamma}(2,2,\dots,2,4)^{\top} \in \R^{2k+2}$ where for $i \in \{1,\ldots, 2k+2\}$, the $i$-th component of $x(k,\gamma)$ corresponds to the node $i-1$ in $H^k$,
In what follows we show that $x(k,\gamma)\in {\LS}_+(H^k)\setminus \stab(H^k)$ for some $\gamma\in (0,1)$ thus proving Theorem \ref{oddanti}.
We first consider $\beta_k = \frac{1}{2k+2}$,  $\gamma\in(0,1)$ and the $(2k+3)\times (2k+3)$ matrix given by

{\small 
\[
Y(k,\gamma) :=
\left [\begin {array}{ccccccccc} (2k+2+\gamma) & 2 & 2 & 2 & 2& 2 & \cdots &
2 & 4\\ \noalign {\medskip}
2 & 2 & 0 & 0 & 0 & 0 & \cdots & 0 & 2\\ \noalign {\medskip}
2 & 0 & 2 & 1-\beta_k & 0 & 0 & \cdots & 0 & 1+\beta_k\\ \noalign {\medskip}
2 & 0 & 1-\beta_k & 2& 1+\beta_k & 0 & \cdots & 0 & 0\\ \noalign {\medskip}
2 & 0 & 0 & 1+\beta_k & 2& 1-\beta_k & \cdots & 0 & 0\\ \noalign {\medskip}
2 & 0 & 0 & 0 & 1-\beta_k & 2& \cdots & 0 & 0\\ \noalign {\medskip}
\vdots & \vdots & \vdots & \vdots & \vdots & \vdots & \ddots & \vdots & \vdots
\\ \noalign {\medskip}
2 & 0 & 0 & 0 & 0 & 0 & \cdots & 2& 1+\beta_k\\ \noalign {\medskip}
4 & 2 & 1+\beta_k & 0 & 0 & 0 & \cdots & 1+\beta_k & 4
\end
 {array}\right ].
 \]
}

\begin{lemma}\label{punto_M+}
For $k\geq 3$, there exists $\gamma\in (0,1)$ such that $Y(k,\gamma)$ is PSD.
\end{lemma}
\begin{proof}
Let us denote by $\tilde{Y}(k)$ the $(2k+2)\times (2k+2)$ submatrix of $Y(k,\gamma)$ obtained after deleting the first row and column.
Also consider $\hat{Y}(k)$, the Schur Complement of the $(1,1)$ entry of $\tilde{Y}(k)$, then
{\small
\[
\hat{Y}(k) =
\left [\begin {array}{ccccccc}
2 & 1-\beta_k & 0 & 0 & \cdots & 0 & 1+\beta_k\\ \noalign {\medskip}
1-\beta_k & 2& 1+\beta_k & 0 & \cdots & 0 & 0\\ \noalign {\medskip}
0 & 1+\beta_k & 2& 1-\beta_k & \cdots & 0 & 0\\ \noalign {\medskip}
0 & 0 & 1-\beta_k & 2& \cdots & 0 & 0\\ \noalign {\medskip}
\vdots & \vdots & \vdots & \vdots & \ddots & \vdots & \vdots
\\ \noalign {\medskip}
0 & 0 & 0 & 0 & \cdots & 2& 1+\beta_k\\ \noalign {\medskip}
1+\beta_k & 0 & 0 & 0 & \cdots & 1+\beta_k & 2
\end
 {array}\right ].
\]
}

\begin{claim}
For every $k \geq 2$, $\tilde{Y}(k)$ is positive definite.
\end{claim}

\begin{proof}
Let us first show that $\hat{Y}(k)$ is positive definite. For this purpose, we only need to verify that every leading principle minor of $\hat{Y}(k)$ is positive.
Let us define $A_0(k):=1$, $B_0(k):=2$ and for $\ell\geq 1$,
{\small
\[
A_{\ell}(k):=\det\left[\begin {array}{ccccccc}
2 & 1-\beta_k & 0 & 0 & \cdots & 0 & 0\\ \noalign {\medskip}
1-\beta_k & 2& 1+\beta_k & 0 & \cdots & 0 & 0\\ \noalign {\medskip}
0 & 1+\beta_k & 2& 1-\beta_k & \cdots & 0 & 0\\ \noalign {\medskip}
0 & 0 & 1-\beta_k & 2& \cdots & 0 & 0\\ \noalign {\medskip}
\vdots & \vdots & \vdots & \vdots & \ddots & \vdots & \vdots
\\ \noalign {\medskip}
0 & 0 & 0 & 0 & \cdots & 2& 1-\beta_k\\ \noalign {\medskip}
0 & 0 & 0 & 0 & \cdots & 1-\beta_k & 2
\end{array}\right],
\]
}
where the matrix in the definition is $2\ell\times2\ell$ and
{\small
\[
B_{\ell}(k):=\det\left[\begin {array}{ccccccc}
2 & 1-\beta_k & 0 & 0 & \cdots & 0 & 0\\ \noalign {\medskip}
1-\beta_k & 2& 1+\beta_k & 0 & \cdots & 0 & 0\\ \noalign {\medskip}
0 & 1+\beta_k & 2& 1-\beta_k & \cdots & 0 & 0\\ \noalign {\medskip}
0 & 0 & 1-\beta_k & 2& \cdots & 0 & 0\\ \noalign {\medskip}
\vdots & \vdots & \vdots & \vdots & \ddots & \vdots & \vdots
\\ \noalign {\medskip}
0 & 0 & 0 & 0 & \cdots & 2& 1+\beta_k\\ \noalign {\medskip}
0 & 0 & 0 & 0 & \cdots & 1+\beta_k & 2
\end{array}\right],
\]
}
where the matrix in the definition is $(2\ell+1)\times(2\ell+1)$.
Using the determinant expansion on $A_{\ell}(k)$ and $B_{\ell}(k)$ we have that for every $\ell\geq 1$,
\[
A_{\ell}(k) = 2 B_{\ell-1}(k) -(1-\beta_k)^2 A_{\ell-1}(k),
\]
\[
B_{\ell}(k) = 2 A_{\ell}(k) -(1+\beta_k)^2 B_{\ell-1}(k),
\]
and
\begin{eqnarray}\label{reccursive}
\det\left(\hat{Y}(k)\right) & = & 2 \left[A_k(k) -(1+\beta_k)^2
B_{k-1}(k) + (1-\beta_k)^k (1+\beta_k)^{k+1}\right] \nonumber\\
& = & 2\left[B_k(k) - A_k (k)+ (1-\beta_k)^k (1+\beta_k)^{k+1}\right].
\end{eqnarray}

Using these recursive formulas, we have that $\hat{Y}(k)$ is positive definite.
Finally, after the Schur Complement Lemma we have that $\tilde{Y}(k)$ is positive definite.
\end{proof}

Using this claim we have:
\begin{claim}
\label{lem:2}
Let $u$ be the (unique) vector such that
\begin{equation}\label{def_u}
\tilde{Y}(k) u =
2 (\one+ \e_{2k+2}).
\end{equation}
Then $Y(k,\gamma)$ is PSD if and only if
$\gamma \geq 1-\beta_k u_{2k+2}$.
\end{claim}

\begin{proof}
Using the Schur Complement Lemma for $\tilde{Y}(k)$
we have that $Y(k,\gamma)$ is PSD if and only if
\[
\tilde{Y}(k) - \frac{4}{2k+2+\gamma}
(\one+ \e_{2k+2})(\one+ \e_{2k+2})^{\top} \text{ is PSD}.
\]
Using the automorphism $\left[\tilde{Y}(k)\right]^{-1/2}
\cdot \left[\tilde{Y}(k)\right]^{-1/2}$ of the PSD cone, the
latter is true if and only if the following matrix
\begin{equation}\label{rank1}
I - \frac{4}{2k+2+\gamma} \left[\tilde{Y}(k)\right]^{-1/2} (\one+ \e_{2k+2})(\one+ \e_{2k+2})^{\top} \left[\tilde{Y}(k)\right]^{-1/2}
\end{equation}
is PSD.
Since
\[
\frac{4}{2k+2+\gamma} \left[\tilde{Y}(k)\right]^{-1/2} (\one+ \e_{2k+2})(\one+ \e_{2k+2})^{\top} \left[\tilde{Y}(k)\right]^{-1/2}
\]
is a rank one matrix, using (\ref{def_u}) we have that the matrix in (\ref{rank1}) is PSD if and only if
\begin{equation}\label{cond1}
1 \geq \frac{4}{2k+2+\gamma}(\one+ \e_{2k+2})^{\top}
[\tilde{Y}(k)]^{-1}
(\one+ \e_{2k+2})=\frac{2(\one+ \e_{2k+2})^{\top} u}{2k+2+\gamma}.
\end{equation}
Now, using the definition of $\tilde{Y}(k)$ we have that
\[
\tilde{Y}(k) \one= 4 (\one)  + (4+ 2 \beta_k )\e_{2k+2}=2\tilde{Y}(k)u+2\beta_k \e_{2k+2},
\]
and then
\[
u=\frac{1}{2} \one - \beta_k [\tilde{Y}(k)]^{-1} \e_{2k+2}.
\]
Therefore,
\begin{align}
2(\one+ \e_{2k+2})^{\top} u & =2(\one+ \e_{2k+2})^{\top} \left( \frac{1}{2} \one - \beta_k [\tilde{Y}(k)]^{-1} \e_{2k+2}\right)\nonumber \\
&=(2k+3) -2 \beta_k (\one+ \e_{2k+2})^{\top} [\tilde{Y}(k)]^{-1}
\e_{2k+2}\nonumber
\end{align}
and again using (\ref{def_u}), we obtain
\begin{equation}\label{u}
2(\one+ \e_{2k+2})^{\top} u = (2k+3) -\beta_k u_{2k+2}.
\end{equation}
Hence, using (\ref{cond1}) and  (\ref{u})
we can conclude that the matrix in (\ref{rank1}) is PSD if and only if
\[1\geq \frac{(2k+3) -\beta_k u_{2k+2}}{2k+2+\gamma}\] or equivalently, if and only if
\[\gamma\geq 1- \beta_k u_{2k+2}.\]
\end{proof}

By the previous claims, to prove that $Y(k,\gamma)$ is PSD for some $\gamma \in (0,1)$, it suffices
to prove that there exists $\gamma \in(0,1)$ such that
\[
\gamma \geq 1-\beta_k u_{2k+2},
\]
where  $u$ is the unique
solution of \eqref{def_u}.
Thus, as long as $u_{2k+2}>0$, we may have $\gamma <1$ as desired.
Using (\ref{u}) we have that
\[
u_{2k+2}=\e_{2k+2}^{\top} u= \frac{1}{2} -\beta_k \e_{2k+2}^{\top} \tilde{Y}(k)^{-1} \e_{2k+2}
\]
and by Cramer's rule and the definitions of $\tilde{Y}$,
$\hat{Y}$, and $A_k$, we conclude
\[
u_{2k+2}=\frac{1}{2}-\beta_k\frac{A_k}{\det\left(\hat{Y}(k)\right)}.
\]
Then, using the recursive formula (\ref{reccursive}) we have that $u_{2k+2}>0$.
This completes the proof.
\end{proof}

Utilizing the previous lemma we are able to prove Theorem \ref{oddanti}.

\begin{proof}[Proof of Theorem \ref{oddanti}]
Recall that we may assume that in $H^k$ the node $2k+1$ is not connected with node $0$.
Let $\gamma\in(0,1)$ and   $x(k,\gamma)=\frac{1}{2k+2+\gamma}(2,2,\dots,2,4)^{\top}\in \R^{2k+2}$ where the $i$-th component of $x(k,\gamma)$ corresponds to node $i-1$ in $H^k$ for $i \in \{1,\ldots,2k+2\}$.
Let $Y^*(k,\gamma)=\frac{1}{2k+2+\gamma}Y(k,\gamma)$. Then, $Y^*(k, \gamma)$ is a symmetric matrix that clearly satisfies that $Y^*(k, \gamma)\e_0=\diag(Y^*(k, \gamma))\in \fra(H^k)$.
Moreover, it is not hard to check that, for $i \in \{1,\dots,2k+2\}$,
\[
Y^*(k, \gamma)\e_i \in \cone(\fra(H^k))\; \text{ and } \; Y^*(k, \gamma)(\e_0-\e_i) \in \cone(\fra(H^k)).
\]
This proves that $Y^*(k, \gamma)\in M(H^k)$. By the previous lemma, there exists $\bar{\gamma}\in (0,1)$ for which $Y^*(k, \bar{\gamma})\in M_+(H^k)$.
Hence, $x(k, \bar{\gamma})\in {\LS}_+(H^k)$.
It only remains to observe that $x(k,\bar{\gamma})$ violates the rank inequality of $H^k$.
Thus, $x(k,\bar{\gamma})\notin \stab(H^k)$.
\end{proof}

\subsection{Operations that preserve ${\LS}_+$-imperfection}

Firstly, we prove Theorem \ref{Thm2.10} on the $k$-stretching operation for $k\geq 1$, already stated in Section \ref{thevalidity}. Actually, we will see that the same proof given in \cite{LipTun2003} for the case $k=0$ can be used for the case $k\geq 1$.
Assume that $\tilde G$ is obtained from $G$ after the $k$-stretching operation on node $v$ and let $u$, $v_1$ and $v_2$ be as in the definition of the operation in Section \ref{operation}. For any $x\in \R^{V(G)}$, we write
$x=\left(\begin{array}{c}
\bar x\\
x_v
\end{array} \right)$ where $\bar x \in \mathbb R^{V(G-v)}$.

For the case $k=0$, the authors in \cite{LipTun2003} prove that if a point $x=\left(\begin{array}{c}
\bar x\\
x_v
\end{array} \right)\in {\LS}_+^r(G)$ then the point $\tilde x$ given by
\[
x_w=\left\{
\begin{array}{ll}
\tilde{x}_w & \text{ if } w\in \{u,v_1,v_2\},\\
\bar{x}_w & \text{ otherwise,}
\end{array}
\right.
\]
satisfies
$\tilde x \in {\LS}_+^r(\tilde G)$. In order to do so they prove that if
\[
Y= \left[\begin{array}{c|ccc|c}
1 & & \bar{x}^{\top} && x_v\\\hline
&&&&\\
\bar{x} & & \bar{X} && \bar{y} \\
&&&&\\\hline
x_v & &\bar{y}^{\top} && x_v
\end{array}
\right] \in M_+({\LS}_+^{r-1}(G))
\]
then
\[
\tilde{Y}=
\left[\begin{array}{c|c|c|c|c}
1 & \bar{x}^{\top} & x_v &  x_v &  \left(1-x_v\right) \\\hline
&&&&\\
\bar{x} & \bar{X} & \bar{y} & \bar{y} & \bar{x}-\bar{y}\\
&&&&\\\hline
x_v & \bar{y}^{\top} & x_v & x_v& 0\\\hline
x_v & \bar{y}^{\top} & x_v & x_v& 0\\\hline
\left(1-x_v\right) & \left(\bar{x}-\bar{y}\right)^{\top} & 0 & 0 & \left(1-x_v\right)
\end{array}
\right] \in M_+({\LS}_+^{r-1}(\tilde G)).
\]
On the other hand, they show that if $\sum_{j \in V(G)} a_j x_j \leq \beta$ is a valid inequality for $\stab(G)$, defining $\tilde{\beta} = \beta +a_v$ and
\[
\tilde{a}_j = \left\{
\begin{array}{lll}
a_v & & \text{ if } j \in \{v_1, v_2, u\},\\
&&\\
a_j & & \textup{otherwise,}
\end{array}
\right.
\]
the inequality $\sum_{j \in V(\tilde G)} \tilde a_j x_j \leq \tilde{\beta}$ is valid for $\stab(\tilde G)$.
Moreover, if $x^*$ violates $\sum_{j \in V(G)} a_j x_j \leq \beta$ then $\tilde{x^*}$ violates $\sum_{j \in V(\tilde G)} \tilde a_j x_j \leq \tilde{\beta}$.

\begin{proof} [Proof of Theorem \ref{Thm2.10}]
It is enough to observe that $\tilde Y\in M_+\left(\LS^{r-1}(\tilde G)\right)$  and the inequality\\
$\sum_{j \in V(\tilde G)} \tilde a_j x_j \leq \tilde{\beta}$ is valid for $\stab(\tilde G)$
even for the case $\tilde{G}$ is obtained after the $k$-stretching on node $v$ in $G$, for $k\geq 1$.
\end{proof}

Let us now consider the clique-subdivision operation defined in \cite{AEF}.
For $x \in\R^n$, let $\bar{x} \in \R^{n+2}$ such that $\bar{x}_i=x_i$ for every $i \in \{1,\ldots,n\}$, $\bar{x}_{n+1}=x_2$ and $\bar{x}_{n+2}=x_1$, and write $\bar{x}=\left(\begin{array}{c}
x\\
x_2\\
x_1
\end{array}\right)$.
In \cite{AEF} the authors prove that if $\tilde G$ be obtained from $G$ by the clique subdivision of the edge $v_1 v_2$ in the clique $K$ and  $x\in {\LS}_+^k(G)$ then $\bar{x}\in \LS^k(\tilde G)$.
In order to do so, they show that if $Y\e_0=\left(
\begin{array}{c}
1 \\ x
\end{array}\right)
$ for
\begin{equation}\label{Y_AEF}
Y= \left[\begin{array}{c|c|c|*{3}{c}}
1   & x_1 & x_2 & & \bar{x}^{\top} & \\\hline
x_1  & x_1 & 0   & & y_1^{\top} & \\\hline
x_2 & 0 & x_2   & & y_2^{\top} & \\\hline
&&&&&\\
\tilde{x} & y_1 & y_2 & & \bar{X} & \\
&&&&&\\
\end{array}
\right] \in M(\LS^{r-1}(G))
\end{equation}
then
\begin{equation}\label{Y_tilde_AEF}
\tilde{Y}=
\left[\begin{array}{c|c|c|*{3}{c}|c|c}
1 & x_1 & x_2 &  & \tilde{x}^{\top} & & x_2 & x_1 \\\hline
x_1 & x_1 & 0 & & y_1^{\top} & & 0 & x_1\\\hline
x_2 & 0 & x_2 & & y_2^{\top} & & x_2 & 0 \\\hline
&&&&&&&\\
\tilde{x} & y_1 & y_2 & & \bar{X} & & y_2 & y_1 \\
&&&&&&&\\\hline
x_2 & 0 & x_2 & & y^{\top}_2 & & x_2 & 0\\\hline
x_1 & x_1 & 0 & & y_1^{\top} & & 0 & x_1
\end{array}
   \right] \in M(\LS^{r-1}(\tilde G)).
\end{equation}

\begin{proof} [Proof of Theorem \ref{cliquesubdivision}]
It is enough to observe that if the matrix $Y$ is PSD then so is the matrix in \eqref{Y_tilde_AEF}. 
\end{proof}

\section{Conclusions and further results}\label{conclu}

In this work, we face the problem of characterizing the stable set polytope of ${\LS}_+$-perfect graphs, a graph class where the Maximum Weight Stable Set Problem is polynomial time solvable. This class strictly includes many
well-known graph classes such as perfect graphs, $t$-perfect graphs, wheels, anti-holes, near-bipartite graphs and the graphs obtained from various suitable compositions of these.
The stable set polytope of either a perfect or a  near-bipartite graph only needs the inequalities associated with the stable set polytopes of its near-bipartite subgraphs.
In a previous work, we have conjectured that the same holds for all ${\LS}_+$-perfect graphs.
In this paper, we prove the validity of this conjecture for $\fs$-perfect graphs, a superclass of near-perfect graphs. Moreover, if $\mathcal{FS}$ denotes the class of $\fs$-perfect graphs, using the definition in (\ref{F}), we actually prove that the conjecture holds for a superclass of $\fs$-perfect graph defined as those graphs for which $\mathcal{FS}(G)=\stab(G)$. Observe that the graph in Figure \ref{clique_sum} satisfies $\mathcal{FS}(G)=\stab(G)$ and it is not $\fs$-perfect.

\begin{figure}[h]
\begin{center}
\includegraphics{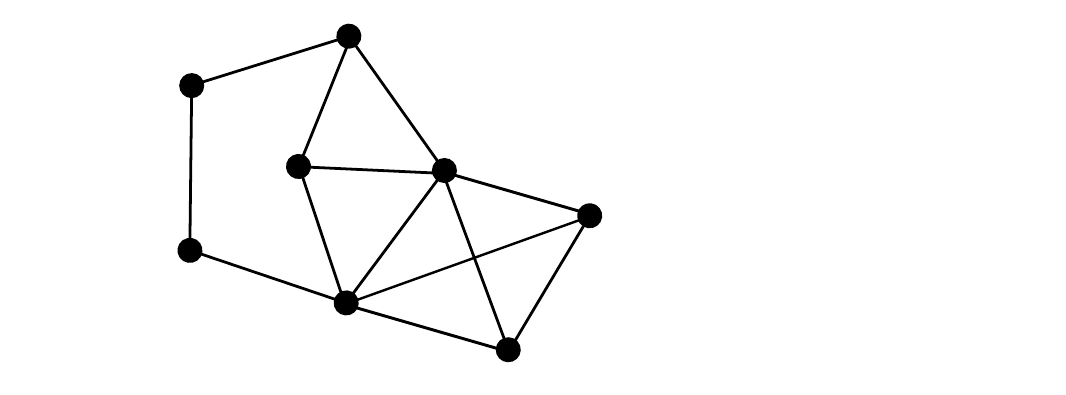}
\caption{A graph $G$ satisfying $\mathcal{FS}(G)=\stab(G)$ which is not $\fs$-perfect. }
\label{clique_sum}
\end{center}
\end{figure}

Also, the results used in the proof of the Theorem \ref{main} allow us to conclude the following:

\begin{corollary}
Let $G$ be a graph such that $V(G)=\{0,1,\ldots,2k+1\}$ with $k\geq 2$ and $G-0$ is minimally imperfect. Then:

\begin{itemize}
\item If $G-0=C_{2k+1}$ then $G$ is ${\LS}_+$-perfect if and only if either $\delta_G(0)\geq 2$ and $G$ has only one odd central cycle or $\delta_G(0)\in \{0,1,2k+1\}$.
\item
If $G-0=\overline{C_{2k+1}}$ with $\alpha(G)=2$ then $G$ is ${\LS}_+$-perfect if and only $\delta_G(0)=2k+1$.
\end{itemize}
\end{corollary}

From the above characterization we identify the forbidden structures in the family of ${\LS}_+$-perfect graphs. In other words,

\begin{corollary}
Let $G$ be an ${\LS}_+$-perfect graph. Then, there is no subgraph $G'$ of $G$ such that
\begin{itemize}
\item $G'-v_0=C_{2k+1}$, $2\leq \delta_{G'}(v_0)\leq 2k$ and $G'$ has at least two odd central cycles, or
\item $G'-v_0=\overline{C_{2k+1}}$ and $k+1 \leq \delta_G'(v_0)\leq 2k$ and $\alpha(G')=2$,
\end{itemize}
for some $v_0\in V(G')$ and $k\geq 2$.
\end{corollary}

\end{document}